\documentclass[12pt,english]{iopart}

\usepackage[usenames,dvipsnames]{color}
\usepackage{graphicx}
\expandafter\let\csname equation*\endcsname\relax
\expandafter\let\csname endequation*\endcsname\relax
\usepackage{iopams}
\usepackage{amsmath}
\usepackage{amssymb,amsthm}
\usepackage{amsfonts}
\usepackage[breaklinks]{hyperref}
\usepackage[utf8]{inputenc}
\usepackage{color}
\usepackage{graphics}
\usepackage{bm}
\usepackage{bbold}
\usepackage{cite}

\newcommand{\bra}[1]{\langle #1|}
\newcommand{\ket}[1]{|#1\rangle}
\newcommand{\braket}[2]{\langle #1|#2\rangle}
\newcommand{\dyad}[1]{|#1\rangle\!\langle#1|}

\newcommand{\expval}[2]{\langle #2 | #1 | #2 \rangle}

\newtheorem{Def}{Definition}
\newtheorem{Rmk}{Remark}
\newtheorem{theorem}{Theorem}
\newtheorem{proposition}{Proposition}
\newtheorem{example}{Example}

\begin{document}

\title{Multiple fidelities and joint numerical range}

\author{Pei Li}%\thanks{616784940@qq.com}
\address{School of Computer Science and Technology, Guangdong University of Technology, China}
%\eads{\mailto{2112305013@gdut.edu.cn}}
\author{Bang-Hai Wang}
\address{School of Computer Science and Technology, Guangdong University of Technology, China}
\eads{\mailto{bhwang@gdut.edu.cn}}

\begin{abstract}
We investigate the effectiveness of entanglement detection based on multiple fidelities via the geometry of the joint
separable numerical range. When all reference states are product states, we derive a necessary and sufficient criterion
for such detection: either some pair of reference states has nontrivial moduli
of the local inner products on both subsystems, or the orthogonal complement of the span of the reference states is
completely entangled.
We further show that there exist sets of reference product states for which no proper subset is effective for
entanglement detection, whereas the full set is. A typical example of this phenomenon is provided by
unextendible product bases. Moreover, for a pair of
reference product states on a bipartite system with arbitrary local dimensions, we
characterize both the joint numerical range and the joint separable numerical range, showing that the joint separable numerical range is
determined solely by their local fidelities, as illustrated by a representative two-qubit example. Our results
offer a systematic approach to designing effective entanglement witnesses and lay the groundwork
for extensions to higher-dimensional and multipartite scenarios.
\bigskip

\end{abstract}

%\maketitle
%%%%%%%%%%%%%%%%%%%%%%%%%%%%%%%%%%%%%%%%%%%%%
\section{Introduction}
\label{sec:int}
%%%%%%%%%%%%%%%%%%%%%%%%%%%%%%%%%%%%%%%%%%%%%
Quantum entanglement\cite{horodecki2009} is a defining feature of quantum mechanics that
has no classical counterpart and has long been a central theme in quantum computation and quantum information\cite{nielsen2010quantum}. Owing to the correlations exhibited by entanglement
that go beyond classical limits, entangled states play a crucial role in a variety of
quantum information processing tasks, such as quantum teleportation\cite{bennett1993},
superdense coding\cite{bennett1992}, and quantum key distribution\cite{bennett2014}.

Entanglement detection remains one of the central challenges in quantum information
theory. While several methods for entanglement detection — such as the positive partial transpose (PPT)
criterion\cite{peres1996}, the computable cross-norm (CCNR) or realignment
criterion\cite{chen2003,rudolph2003}, the range criterion\cite{horodecki1997},
and entanglement witnesses (EWs)\cite{chruscinski2014} — have been proposed,
each has inherent limitations. Currently, EWs are
widely used in experiments to detect the entanglement of unknown quantum states\cite{guhne2009entanglement}.
A Hermitian operator $W$ is called an entanglement witness if it
satisfies the following two conditions: \rm (i) For all separable states $\rho_{\text{sep}}$,
 $
 \mathrm{Tr}(W \rho_{\text{sep}}) \ge 0,
 $
\rm (ii) There exists at least one entangled state $\rho_{\text{ent}}$ such that
 $
 \mathrm{Tr}(W \rho_{\text{ent}}) < 0.
 $
In other words, a quantum state is confirmed to be entangled if it has a negative
expectation value on some entanglement witness.
 Many methods have been proposed for constructing EWs\cite{terhal2001family,lewenstein2000optimization,toth2005entanglement,piani2007class}.
A particularly practical class of EWs used in experiments is the family of fidelity-based witnesses, which rely
on the observation that any quantum state whose fidelity with a given entangled
pure state is sufficiently high must itself be entangled\cite{bourennane2004experimental}. Such witnesses can be written as
\begin{equation}
\label{fidelitybased}
 W=\alpha I-\dyad{\psi},
\end{equation}
where $\dyad{\psi}$ denotes a fixed entangled pure state and $\alpha$
is the maximum fidelity between $\dyad{\psi}$ and separable states. Despite
their operational simplicity and experimental accessibility, fidelity-based
entanglement witnesses do not detect all entangled states.
Those entangled states that cannot be detected by any fidelity-based witness are
known as unfaithful entangled states\cite{weilenmann2020,guhne2021}.

% \begin{equation}
% F(\rho)=\max_{U_{A}, U_{B}}\mathrm{tr}(\rho U_{A}\otimes U_{B}\Phi_{1}U_{A}^{\dagger}\otimes U_{B}^{\dagger}),
% \end{equation}

To address the limitations of fidelity-based witnesses based on a single reference state, Zhang et al. \cite{zhang2025}
 proposed detecting entanglement by measuring multiple fidelities with respect to several
reference states, defined by
\begin{equation}
 \mathrm{F}(\rho, (\dyad{\psi_1}, \ldots, \dyad{\psi_k}))=(\expval{\rho}{\psi_1}
, \ldots, \expval{\rho}{\psi_k}),
\label{k-tuple}
\end{equation}
where $\rho$ denotes the target quantum state and the $\dyad{\psi_i}$ are reference
pure states, which can be either entangled or product states. If the resulting tuple does
not lie within the set of all possible values attainable by separable states, then
$\rho$ is certified as entangled.
Compared with fidelity-based witnesses built from a single reference state, this
approach has greater detection power and is not tied to a single fixed entangled
reference state. A key issue is therefore how to choose the reference states
$\dyad{\psi_i}$ such that the entanglement detection method based on Eq.~\eqref{k-tuple}
is effective. This can be naturally addressed using a geometric perspective based on the joint numerical range (JNR)
and the joint separable numerical range (JSNR) \cite{gutkin2013}. The resulting geometric
framework generalizes conventional fidelity-based witnesses and provides a systematic way
to analyze which choices of $\dyad{\psi_i}$ lead to effective entanglement detection. By
using the supporting-hyperplane description of the JSNR, one
can identify choices of reference states whose associated witness operators detect
entanglement, including certain PPT entangled states\cite{bennett1999}. This
geometric viewpoint leads to a unified framework for entanglement detection based on
multiple fidelities and JSNR geometry.

In \cite{wu2020}, Wu et al. established a sufficient condition for two observables to be effective in entanglement detection.
In this work, we demonstrate that when the observables are rank-one projectors onto
product states, this condition is both necessary and sufficient. Furthermore,
we derive the general necessary and sufficient conditions for a set of reference product states to be effective for entanglement detection.
In particular, we
demonstrate that for a pair of reference product states on a bipartite system with arbitrary local dimensions,
the geometry of the JSNR is governed solely by their local fidelities, which allows us to construct the corresponding
JSNR from local data in a representative two-qubit example.

The remainder of this paper is structured as follows. Section \ref{sec:pre} reviews
the mathematical preliminaries on the JNR and the JSNR, together with the
support-function formulation and the relevant notions of CES and CSS. Section \ref{sec:res} presents our main
theoretical results, providing the necessary and sufficient conditions for when
two reference product states, and more generally $k$ reference product states, can detect entanglement, followed by a discussion on their local
unitary invariance and analytical examples. Finally, Section \ref{sec:con} summarizes
our work and discusses potential extensions to multipartite and higher-dimensional
scenarios.

\section{Preliminaries}
\label{sec:pre}
In this section, we introduce several notions and mathematical tools that will be used
throughout the paper.
 Throughout the paper, for two pure states \(|\phi\rangle\) and \(|\psi\rangle\), we
write \( |\langle\phi|\psi\rangle|^2 \) for their fidelity.
For two product states \(|a_1\rangle\otimes|b_1\rangle\) and \(|a_2\rangle\otimes|b_2\rangle\),
we write \(c_A:=|\langle a_1|a_2\rangle|^2\) and \(c_B:=|\langle b_1|b_2\rangle|^2\)
for the corresponding local fidelities.
 \subsection{Joint numerical range}
Consider $k$ Hermitian operators $A_{1}, \ldots, A_{k}$ acting on a Hilbert space of
dimension $d_{A} d_{B}$. For a set \(\mathcal X\) of quantum states and a tuple
\(A=(A_1,\ldots,A_k)\) of Hermitian operators, define the restricted joint numerical
range by\cite{simnacher2021confident}
\begin{equation}
 \label{eq:JSNR}
 L_{\mathcal X}(A_1, \ldots, A_k)=\{(\mathrm{Tr}(\rho A_1), \ldots, \mathrm{Tr}(\rho A_k))
 \in\mathbb{R}^k\mid\rho\in \mathcal X\}.
\end{equation}
When \(\mathcal X\) is the set of all quantum states, \(L_{\mathcal X}(A_1,\ldots,A_k)\) is the joint numerical
range, and we simply write $L(A_1,\ldots,A_k)$. When \(\mathcal X\) is the set of separable
states, we obtain the joint separable numerical range, denoted by
$L_{\mathrm{SEP}}(A_1,\ldots,A_k)$, which characterizes all $k$-tuple expectations
 attainable by separable states.

As long as \(\mathcal X\) is a compact convex set, \(L_{\mathcal X}(A_1,\ldots,A_k)\) is also compact and convex,
because it is the image of \(\mathcal X\) under the linear map
\(\rho \mapsto (\Tr(\rho A_1),\ldots,\Tr(\rho A_k))\). If we let $A_i = \dyad{\psi_i}$,
a quantum state $\rho$ is certified as entangled whenever the tuple in
 Eq.~\eqref{k-tuple} lies outside $L_{\mathrm{SEP}}(A_1,\ldots,A_k)$.
For the tuple of operators $A = (A_1, \dots, A_k)$, we abbreviate
$L_{\mathrm{SEP}}(A):=L_{\mathrm{SEP}}(A_1,\ldots,A_k)$. Its support function
is defined as\cite{bertsekas2003convex}
\begin{equation}
h_{L_{\mathrm{SEP}}(A)} (\mathbf{n}):= \max_{\rho_{\mathrm{sep}}} \sum_{i=1}^{k} n_i \, \mathrm{Tr}(\rho_{\mathrm{sep}} A_i),
\label{eq:support_function}
\end{equation}
where $\mathbf{n} = (n_1, \dots, n_k) \in \mathbb{R}^k$ denotes the direction of support,
and $\rho_{\mathrm{sep}}$ ranges over all separable states.
Since a compact convex set is uniquely determined by its support function,
 $L_{\mathrm{SEP}}(A)$ can be expressed as
\begin{equation}
L_{\mathrm{SEP}} (A) = \bigcap_{\mathbf{n}\in \mathbb{R}^k}
\Bigl\{ \mathbf{x} \in \mathbb{R}^k \ \big| \ \mathbf{n}\cdot \mathbf{x} \le h_{L_{\mathrm{SEP}}(A)} (\mathbf{n}) \Bigr\},
\label{eq:L_SEP_support}
\end{equation}
 The corresponding detection region in the space of expectation tuples is
 \begin{equation}
 \mathcal{D}_A:= \bigcup_{\mathbf{n} \in \mathbb{R}^k}
\Bigl\{ \mathbf{x} \in L(A_1,\ldots,A_k) \ \big| \ \mathbf{n}\cdot \mathbf{x} > h_{L_{\mathrm{SEP}}(A)} (\mathbf{n}) \Bigr\},
\label{eq:DA_support_union}
\end{equation}
 A state $\rho$ is detected if and only if its expectation tuple lies in $\mathcal{D}_A$.
Equivalently, the set of detected states can be expressed as
\begin{equation}
D_A:= \bigcup_{\mathbf{n} \in \mathbb{R}^k}
\Bigl\{ \rho \ \big| \ \mathrm{Tr}\Big(\rho \, W_{\mathbf{n}}\Big) < 0 \Bigr\},
\quad
W_{\mathbf{n}}:= h_{L_{\mathrm{SEP}}(A)} (\mathbf{n})I - \sum_{i=1}^{k} n_i A_i.
\label{eq:DA_EW_negative}
\end{equation}
It is straightforward to see that a set of observables $\{A_i\}$ can be used to detect entanglement
if and only if there exists a direction $\mathbf{n}$ such that the corresponding witness
$W_{\mathbf{n}}$ in Eq.~\eqref{eq:DA_EW_negative} is an effective entanglement witness.
Geometrically, Eq.~\eqref{eq:DA_EW_negative} detects entanglement by finding a supporting
hyperplane of the JSNR with outward normal vector $\mathbf{n}$.
To analyze when such witnesses are effective, we next recall the notions of CES and CSS.

 \subsection{Completely Entangled Subspace and Completely Separable Subspace}
 A subspace $\mathcal{S} \subset \mathcal{H}_A \otimes \mathcal{H}_B$ is called a completely entangled subspace (CES) \cite{parthasarathy2004} if it contains no product state, i.e.,
\[
|\psi_A\rangle \otimes |\psi_B\rangle \notin \mathcal{S}, \quad \forall\, |\psi_A\rangle \in \mathcal{H}_A, \, |\psi_B\rangle \in \mathcal{H}_B.
\]
In other words, every vector in $\mathcal{S}$ is entangled. If a mixed state is supported entirely on a CES, it must be entangled.
The maximal possible dimension of a CES in a $d_A \times d_B$ system is\cite{parthasarathy2004}
\begin{equation}
 \label{eq:CESmax}
 \dim(\mathcal{S}_{\max}) = (d_A - 1)(d_B - 1).
\end{equation}

 This notion is closely related to unextendible product bases (UPBs)\cite{bennett1999}:
the orthogonal complement of a UPB is a CES.
Similarly, we can define the concept of a completely separable subspace.
\begin{Def}
A subspace $\mathcal{T} \subset \mathcal{H}_A \otimes \mathcal{H}_B$ is called a
completely separable subspace (CSS),
if every vector in it is a product state, i.e.,
\[
|\phi\rangle \in \mathcal{T} \Rightarrow |\phi\rangle = |\psi_A\rangle \otimes |\psi_B\rangle.
\]
\end{Def}
A CSS is therefore a subspace consisting entirely of product vectors.
As will be proved in Proposition \ref{pro:two product}, in the bipartite setting, if a subspace contained two product vectors with
nonproportional local factors on both subsystems, then their span would already
contain entangled vectors. Hence, in a CSS, all vectors must share a common local
factor on at least one subsystem; that is, either every vector has the same
\(\mathcal{H}_A\)-side factor or every vector has the same \(\mathcal{H}_B\)-side factor.
Therefore a CSS has the form
$|a\rangle \otimes \mathcal{T}_B$ or $\mathcal{T}_A \otimes |b\rangle$,
for some subspaces \(\mathcal{T}_A \subseteq \mathcal{H}_A\) and
\(\mathcal{T}_B \subseteq \mathcal{H}_B\).
Consequently, the maximal possible dimension of a CSS is
\begin{equation}
\dim(\mathcal{T}_{\max}) = \max(d_A, d_B).
\end{equation}

\section{Main results}
\label{sec:res}
As discussed in the Introduction, choosing effective reference states for multiple fidelities is not straightforward.
Clearly, if one of the reference states is entangled, it can certainly be used to detect
entanglement, because measuring the fidelity between the target state and an entangled pure
state is sufficient to detect certain entangled states.

 The key question is under what conditions a set of reference product states
$\{\ket{\psi_i}\}_{i=1}^k$
can be effective for entanglement detection.
We note that for each fixed direction vector $\mathbf{n}$ of the support function,
Eq.~\eqref{eq:DA_EW_negative} has the form
\begin{equation}
\label{eq:single}
W = \alpha I - A, \qquad
\text{with } \alpha = \max_{\sigma \in \mathrm{SEP}} \Tr(\sigma A).
\end{equation}
The effectiveness of Eq.~\eqref{eq:single} depends entirely on $A$. For a single
observable \(A\), we have the following conclusion.

\begin{proposition}
 \label{pro:obs}
 An observable \(A\) yields an effective entanglement witness of the form
 Eq.~\eqref{eq:single} if and only if the eigenspace corresponding to the maximal eigenvalue of \(A\) is a
 CES.
\end{proposition}
\begin{proof}
 Let \(E_{\max}(A)\) denote the eigenspace corresponding to the maximal eigenvalue
 \(\lambda_{\max}(A)\) of \(A\).

 For sufficiency, assume that \(E_{\max}(A)\) is a CES. Since the set of separable
 states is compact and the map \(\sigma \mapsto \Tr(\sigma A)\) is continuous, the value
 \(\alpha=\max_{\sigma\in\mathrm{SEP}}\Tr(\sigma A)\) is attained by some separable state.
 If \(\alpha=\lambda_{\max}(A)\), then an optimizing separable state \(\sigma_\ast\) would satisfy
 \(\Tr(\sigma_\ast A)=\lambda_{\max}(A)\), which is possible only if \(\sigma_\ast\) is supported on
 \(E_{\max}(A)\). This contradicts the fact that \(E_{\max}(A)\) is a CES. Hence
 \[
 \alpha<\lambda_{\max}(A).
 \]
 By the discussion in Section \ref{sec:pre}, any state supported on \(E_{\max}(A)\)
 is entangled. Therefore, for any state \(\rho\) supported on \(E_{\max}(A)\), we have
 \(\Tr(A\rho)=\lambda_{\max}(A)\), and thus
 \[
 \Tr(W\rho)=\alpha-\lambda_{\max}(A)<0.
 \]
 On the other hand, by the definition of \(\alpha\), \(\Tr(W\sigma)\ge 0\) for all
 separable states \(\sigma\). Therefore \(W\) is an effective entanglement witness.

 For necessity, assume that \(E_{\max}(A)\) contains a separable state \(\sigma_{\max}\).
 Then
 \[
 \Tr(A\sigma_{\max})=\lambda_{\max}(A),
 \]
 so \(\alpha=\lambda_{\max}(A)\). Therefore, for any state \(\rho\),
 \[
 \Tr(W\rho)=\alpha-\Tr(A\rho)\ge \lambda_{\max}(A)-\lambda_{\max}(A)=0.
 \]
 Hence \(W\) has no negative expectation value on any state, so Eq.~\eqref{eq:single}
 does not define an entanglement witness.
\end{proof}
Equivalently, Eq.~\eqref{eq:single} is effective if and only if
\(\alpha<\lambda_{\max}(A)\), which holds precisely when \(E_{\max}(A)\) is a CES.

As discussed in Section \ref{sec:pre}, a set of \(k\) observables detects entangled
states via the JSNR if and only if at least one witness in the family
 \eqref{eq:DA_EW_negative} is effective. Since directly analyzing the eigenvectors of the space spanned
by \(k\) product-state projectors is very challenging, we first consider the case of
two product states,
\begin{proposition}
 \label{pro:two product}
For any two product states \(|\psi_1\rangle = |a_1\rangle \otimes |b_1\rangle\) and \(|\psi_2\rangle = |a_2\rangle \otimes |b_2\rangle\), if they do not share a common local factor on either subsystem (i.e., \(|a_1\rangle \not\propto |a_2\rangle\) and \(|b_1\rangle \not\propto |b_2\rangle\)), then any nontrivial linear combination
\[
|\phi\rangle = \alpha |\psi_1\rangle + \beta |\psi_2\rangle, \quad \alpha, \beta \neq 0,
\]
is entangled.
\end{proposition}

\begin{proof}
Suppose, for contradiction, that \(|\phi\rangle\) is separable.
 Choose bases \(\{|e_A\rangle,|f_A\rangle\}\) and \(\{|e_B\rangle,|f_B\rangle\}\) for the two-dimensional subspaces
\(\mathrm{span}\{|a_1\rangle,|a_2\rangle\}\) and \(\mathrm{span}\{|b_1\rangle,|b_2\rangle\}\), respectively, such that
 \[
 |a_1\rangle = | e_A \rangle, \qquad |a_2\rangle = s_A | e_A \rangle + t_A | f_A \rangle,
 \]
 \[
|b_1\rangle = |e_B\rangle, \qquad |b_2\rangle = s_B |e_B\rangle + t_B |f_B\rangle,
\]
where \(t_A \neq 0\) and \(t_B \neq 0\), since otherwise \(|a_2\rangle\) would be proportional to \(|a_1\rangle\) or \(|b_2\rangle\) would be proportional to \(|b_1\rangle\), contradicting the assumptions. Then
\[
\begin{aligned}
|\phi\rangle
&=
\alpha |e_A e_B\rangle
+ \beta (s_A |e_A\rangle + t_A |f_A\rangle)(s_B |e_B\rangle + t_B |f_B\rangle) \\
&=
(\alpha + \beta s_A s_B)|e_A e_B\rangle
+ \beta s_A t_B |e_A f_B\rangle
+ \beta t_A s_B |f_A e_B\rangle
+ \beta t_A t_B |f_A f_B\rangle.
\end{aligned}
\]
If \(|\phi\rangle\) were a product state, then its coefficient matrix in this \(2\times 2\) basis
would have rank one, so its determinant would vanish. However, the corresponding matrix is
\[
\begin{pmatrix}
\alpha + \beta s_A s_B & \beta s_A t_B \\
\beta t_A s_B & \beta t_A t_B
\end{pmatrix},
\]
whose determinant equals
\[
(\alpha + \beta s_A s_B)(\beta t_A t_B) - (\beta s_A t_B)(\beta t_A s_B)
= \alpha \beta t_A t_B \neq 0.
\]
This contradiction shows that \(|\phi\rangle\) cannot be separable.
Therefore, \(|\phi\rangle\) must be entangled.
\end{proof}
\begin{Rmk}
 In algebraic geometry, the subspace spanned by any two product states
 intersects the Segre variety in at most two points, unless they share
 a common local factor on one subsystem \cite{bengtsson2017}.
\end{Rmk}

%We now know the necessary and sufficient condition for a single observable to detect entangled states.
By combining Eq.~\eqref{eq:DA_EW_negative} with Proposition \ref{pro:obs}, we
 can directly obtain the conditions under which measurements of multiple fidelities
 can effectively detect entanglement.

\begin{theorem}
 \label{th1}
Two linearly independent product states $\ket{\psi_1} = \ket{\psi_{A_1}}\otimes\ket{\psi_{B_1}}$ and
$\ket{\psi_2}=\ket{\psi_{A_2}}\otimes\ket{\psi_{B_2}}$ can be used to detect entangled
 states via multiple fidelities if and only if the moduli of their local inner products satisfy
 \[
 0 < |\braket{\psi_{A_1}}{\psi_{A_2}}| < 1
 \quad \text{and} \quad
 0 < |\braket{\psi_{B_1}}{\psi_{B_2}}| < 1.
 \]
 \end{theorem}
\begin{proof}
 For sufficiency, let
 \[
 c:=\braket{\psi_1}{\psi_2}.
 \]
 Since
 \[
 c=\braket{\psi_{A_1}}{\psi_{A_2}}\,\braket{\psi_{B_1}}{\psi_{B_2}}
 \]
 and
 \[
 0 < |\braket{\psi_{A_1}}{\psi_{A_2}}| < 1,
 \qquad
 0 < |\braket{\psi_{B_1}}{\psi_{B_2}}| < 1,
 \]
 we have \(0<|c|<1\). Consider the observable
 \[
 M=\dyad{\psi_1}+\dyad{\psi_2}.
 \]
 Let \(\eta:=\overline{c}/|c|\). Then a direct calculation shows that
 \[
 M(\ket{\psi_1}\pm \eta \ket{\psi_2})=(1\pm |c|)(\ket{\psi_1}\pm \eta \ket{\psi_2}).
 \]
 Hence the maximal eigenvalue of \(M\) is \(1+|c|\), and its eigenspace is one-dimensional,
 spanned by
 \[
 \ket{\phi_{\max}}=\ket{\psi_1}+\eta \ket{\psi_2}.
 \]
 Since both coefficients in \(\ket{\phi_{\max}}\) are nonzero, Proposition \ref{pro:two product}
 implies that \(\ket{\phi_{\max}}\) is entangled. Therefore the eigenspace corresponding
 to the maximal eigenvalue of \(M\) is a CES. By Proposition \ref{pro:obs},
 \(\ket{\psi_1}\) and \(\ket{\psi_2}\) can be used to detect entanglement.

 For necessity, let
 \[
 M=n_1\dyad{\psi_1}+n_2\dyad{\psi_2}
 \]
 be an arbitrary Hermitian linear combination.
 First assume that \(\ket{\psi_1}\) and \(\ket{\psi_2}\) are orthogonal. Then
 \(\ket{\psi_1}\) and \(\ket{\psi_2}\) are eigenvectors of \(M\) with eigenvalues
 \(n_1\) and \(n_2\), respectively, and
 \[
 \lambda_{\max}(M)=\max\{n_1,n_2,0\}.
 \]
 If \(\lambda_{\max}(M)>0\), then the maximal eigenspace contains either
 \(\ket{\psi_1}\), \(\ket{\psi_2}\), or both, so it is not a CES.
 If \(\lambda_{\max}(M)=0\), then \(\ker M\) has codimension at most \(2\), hence
 \[
 \dim \ker M \ge d_A d_B -2 > (d_A-1)(d_B-1),
 \]
 where the strict inequality holds for all \(d_A,d_B\ge 2\). By Eq.~\eqref{eq:CESmax},
 \(\ker M\) cannot be a CES. Therefore, for orthogonal \(\ket{\psi_1}\) and
 \(\ket{\psi_2}\), the maximal eigenspace of \(M\) is never a CES.

 Now assume that \(\ket{\psi_{A_1}} \propto \ket{\psi_{A_2}}\) or
 \(\ket{\psi_{B_1}} \propto \ket{\psi_{B_2}}\). Then
 \(\mathrm{span}\{\ket{\psi_1},\ket{\psi_2}\}\) is a CSS. Every nonzero eigenspace of \(M\)
 is contained in this span, and hence contains only product states. If instead
 \(\lambda_{\max}(M)=0\), then again \(\ker M\) has codimension at most \(2\), so it
 cannot be a CES by the same dimension argument above. Thus the maximal eigenspace of
 \(M\) is never a CES in this case either.

 In all cases where the local conditions in the theorem fail, no Hermitian linear
 combination \(M=n_1\dyad{\psi_1}+n_2\dyad{\psi_2}\) has a maximal eigenspace that is a
 CES. By Proposition \ref{pro:obs}, \(\ket{\psi_1}\) and \(\ket{\psi_2}\) cannot be
 used to detect entanglement.
 \end{proof}
\begin{Rmk}
 The sufficiency of Theorem \ref{th1} can also be derived naturally from \cite{wu2020}, which states that two product observables
\( A_1 \otimes B_1 \) and \( A_2 \otimes B_2 \) are effective for entanglement detection provided that
\( A_1 \) and \( A_2 \) share no common eigenvectors, and neither do \( B_1 \) and \( B_2 \).

\end{Rmk}

Theorem \ref{th1} can be extended to the case of $k > 2$, but there are subtle differences.
For $k > 2$, it is not always the case that the eigenspace corresponding to the zero
eigenvalue contains separable states. With slight modifications, we can obtain the following result.

\begin{theorem}
 \label{th2}
 $k$ linearly independent product states $\ket{\psi_1} = \ket{\psi_{A_1}}\otimes\ket{\psi_{B_1}}$,
 $\ket{\psi_2}=\ket{\psi_{A_2}}\otimes\ket{\psi_{B_2}}, \ldots,
 \ket{\psi_k}=\ket{\psi_{A_k}}\otimes\ket{\psi_{B_k}}$ can be used to detect entangled
 states via multiple fidelities if and only if either there exist \(i \neq j\) such that the
 moduli of their local inner products satisfy
 \[
 0 < |\braket{\psi_{A_i}}{\psi_{A_j}}| < 1
 \quad \text{and} \quad
 0 < |\braket{\psi_{B_i}}{\psi_{B_j}}| < 1,
 \]
 or the orthogonal complement of $\mathrm{span}\{\ket{\psi_i}, i\in[k]\}$
 forms a CES.
\end{theorem}
\begin{proof}
 We first prove sufficiency.
 If there exist \(i \neq j\) such that
 \(0 < |\braket{\psi_{A_i}}{\psi_{A_j}}| < 1\) and
 \(0 < |\braket{\psi_{B_i}}{\psi_{B_j}}| < 1\), then
 Theorem \ref{th1} shows that the pair
 \(\{\ket{\psi_i},\ket{\psi_j}\}\) already detects entanglement, and therefore
 the full set does as well. This proves sufficiency in the first case.

 It remains to prove sufficiency in the second case, namely when no pair satisfies
 the first alternative and
 \[
 K:=\mathrm{span}\{\ket{\psi_i}: i\in[k]\}^{\perp}
 \]
 is a CES. Then, for every \(i \neq j\), if \(\braket{\psi_i}{\psi_j}\neq 0\),
 Theorem \ref{th1} implies that either
 \(\ket{\psi_{A_i}} \propto \ket{\psi_{A_j}}\) or
 \(\ket{\psi_{B_i}} \propto \ket{\psi_{B_j}}\).

 Define an equivalence relation on \([k]\) by declaring \(i \sim j\) if either
 \(i=j\) or there exists a finite sequence
 \[
 i=i_0,i_1,\ldots,i_r=j
 \]
 such that
 \[
 \braket{\psi_{i_t}}{\psi_{i_{t+1}}}\neq 0,\qquad t=0,\ldots,r-1.
 \]
 Let \(J_1,\ldots,J_l\) be the corresponding equivalence classes. By construction,
 if \(r\neq s\) and \(i\in J_r\), \(j\in J_s\), then
 \(\braket{\psi_i}{\psi_j}=0\); otherwise \(i\) and \(j\) would belong to the same
 class.

 We claim that within each class \(J_r\), all product states share a common local
 factor on the same subsystem. Indeed, suppose there exist three states in the same
 class such that \(\ket{\psi_i}\) and \(\ket{\psi_j}\) share an \(A\)-side factor,
 while \(\ket{\psi_j}\) and \(\ket{\psi_k}\) share a \(B\)-side factor. Then we may write
 \[
 \ket{\psi_i}=\ket{a}\otimes\ket{b_i},\qquad
 \ket{\psi_j}=\ket{a}\otimes\ket{b_j},\qquad
 \ket{\psi_k}=\ket{a_k}\otimes\ket{b_j},
 \]
 where \(\braket{b_i}{b_j}\neq 0\) and \(\braket{a}{a_k}\neq 0\) because both adjacent
 pairs are nonorthogonal. Consequently,
 \[
 \braket{\psi_i}{\psi_k}
 =
 \braket{a}{a_k}\braket{b_i}{b_j}\neq 0.
 \]
 Hence \(\ket{\psi_i}\) and \(\ket{\psi_k}\) must also share a local factor. If they
 shared the \(A\)-side factor, then \(\ket{a_k}\propto\ket{a}\), so
 \(\ket{\psi_k}\propto\ket{\psi_j}\), contradicting linear independence. If they
 shared the \(B\)-side factor, then \(\ket{b_i}\propto\ket{b_j}\), so
 \(\ket{\psi_i}\propto\ket{\psi_j}\), again a contradiction. Therefore the type of
 shared local factor cannot change along a nonorthogonal chain. It follows that each
 class \(J_r\) is of exactly one of the two forms
 \[
 \ket{\psi_{j_m}}=\ket{a_r}\otimes\ket{b_{j_m}}, \qquad j_m\in J_r,
 \]
 or
 \[
 \ket{\psi_{j_m}}=\ket{a_{j_m}}\otimes\ket{b_r}, \qquad j_m\in J_r.
 \]

 For an arbitrary Hermitian linear combination
 \begin{equation}
 \label{eq:M}
 M =\sum_{i=1}^{k} n_i \dyad{\psi_i} = \sum_{r=1}^{l} G_r,
 \end{equation}
 define \(G_r\) by
 \begin{equation}
 \label{local:a}
 G_r = \sum_{j_m\in J_r} n_{j_m} \, |a_{r}\rangle \langle a_{r}| \otimes |b_{j_m}\rangle \langle b_{j_m}|,
 \end{equation}
 in the first case, or by
 \begin{equation}
 G_r = \sum_{j_m\in J_r} n_{j_m} \, |a_{j_m}\rangle \langle a_{j_m}| \otimes |b_{r} \rangle \langle b_{r} |
 \end{equation}
 in the second case. In either case, the range of \(G_r\) is contained in a CSS, namely
 \begin{equation}
 \label{eq:CSS}
 \mathrm{range}(G_r)= \ket{a_r} \otimes \mathcal{T}_{B,r}
 \quad \text{or} \quad
 \mathrm{range}(G_r) = \mathcal{T}_{A,r}\otimes \ket{b_r}.
 \end{equation}
 Since different classes are pairwise orthogonal, we have
 \begin{equation}
 \label{eq:oth}
 G_r G_s = 0,\qquad r\neq s.
 \end{equation}
 Therefore the support spaces of the \(G_r\) are mutually orthogonal, and a spectral
 decomposition of each \(G_r\),
 \begin{equation} \label{eq:Gj_spectral}
 G_r = \sum_{\alpha} \lambda_{r,\alpha} \, |\psi_{r,\alpha} \rangle \langle \psi_{r,\alpha} |,
 \end{equation}
 combines into a spectral decomposition of \(M\):
 \begin{equation}\label{eq:M_decompose}
 M = \sum_{r=1}^l \sum_{\alpha} \lambda_{r,\alpha} \, |\psi_{r,\alpha}\rangle \langle \psi_{r,\alpha}|.
 \end{equation}
 Every eigenvector of \(G_r\) corresponding to a nonzero eigenvalue lies in
 \(\mathrm{range}(G_r)\), hence is a product state by Eq.~\eqref{eq:CSS}. It follows
 that any eigenspace of \(M\) corresponding to a nonzero eigenvalue is spanned by
 product states and therefore cannot be a CES.

 Then \(K\subseteq \ker M\) for every choice of the coefficients \(n_i\).
 Since \(K\) is a CES, choose \(n_i<0\) for all \(i\). Then \(M\) is negative semidefinite,
 so its maximal eigenvalue is \(0\). Moreover, \(M\) is strictly negative on
 \(\mathrm{span}\{\ket{\psi_i}\}\): indeed, for any nonzero
 \(\ket{\phi}\in \mathrm{span}\{\ket{\psi_i}\}\),
 \[
 \bra{\phi} M \ket{\phi}
 = \sum_{i=1}^k n_i |\braket{\psi_i}{\phi}|^2 <0,
 \]
 because \(\braket{\psi_i}{\phi}=0\) for all \(i\) would force
 \(\ket{\phi}\) to be orthogonal to \(\mathrm{span}\{\ket{\psi_i}\}\) while also
 lying in it, hence \(\ket{\phi}=0\). Therefore \(\ker M=K\), so the maximal
 eigenspace of \(M\) is exactly \(K\), which is a CES. By Proposition \ref{pro:obs},
 the set \(\{\ket{\psi_i}\}_{i=1}^k\) can be used to detect entanglement. This proves
 sufficiency.

 We now prove necessity by contraposition. Assume that no pair satisfies the first
 alternative and that \(K\) is not a CES. Then \(K\) contains a product state.
 For an arbitrary Hermitian linear combination \(M\), the argument above still shows
 that every eigenspace corresponding to a nonzero eigenvalue is spanned by product
 states and therefore cannot be a CES. If \(\lambda_{\max}(M)=0\), then the maximal
 eigenspace is \(\ker M\), which contains \(K\), and therefore also contains a
 product state. Thus the maximal eigenspace of \(M\) is never a CES. By Proposition
 \ref{pro:obs}, no effective witness can be obtained from \(\{\ket{\psi_i}\}_{i=1}^k\),
 so the set cannot be used to detect entanglement. This proves necessity.
 \end{proof}

\begin{Rmk}
Theorem \ref{th2} shows that a set of \(k\) product-state projectors can be effective
only collectively: no proper subset is capable of detecting entanglement, while the full set is effective. A typical example is provided by a UPB.
The projectors onto the UPB vectors form a family of mutually orthogonal projectors, and any proper
subset fails to produce an effective EW because the eigenspaces corresponding to their nonzero
eigenvalues always contain product states.
However, when all projectors are used together, the eigenspace associated with
the zero eigenvalue becomes a CES. This implies that if the multiple
fidelities of a target state with respect to a UPB lie at the origin,
the state is supported on the UPB-complement subspace and is therefore entangled.
The normalized projector onto this complement is the standard PPT entangled state associated with the UPB \cite{bennett1999}.

\end{Rmk}

Unlike the JNR (or JSNR) of general observables\cite{wu2020},
the geometric structure of the JNR (or JSNR) for pure states may depend only
on a small number of parameters. One important reason is that both the JNR and the JSNR
are invariant under local unitary transformations.
\begin{proposition}[Local unitary invariance of \( L_{\mathcal X}\)]
\label{prop:LU-invariance-LX}
Let \(\mathcal X\) be a set of quantum states on a bipartite Hilbert space
$\mathcal{H}=\mathcal{H}_A\otimes\mathcal{H}_B$,
and let $A_1,\dots,A_k$ be Hermitian operators acting on $\mathcal{H}$.
Assume that \(\mathcal X\) is invariant under local unitary conjugations, namely
\begin{equation}
\rho \in \mathcal X
\ \Longrightarrow \
(U_A\otimes U_B)\,\rho\,(U_A\otimes U_B)^\dagger \in \mathcal X
\label{eq:X-LU-invariant}
\end{equation}
for all local unitary operators $U_A\otimes U_B$.
Then the set
\begin{equation}
L_{\mathcal X}(A_1,\dots,A_k)
=
\Bigl\{
\bigl(
\Tr(\rho A_1),\dots,\Tr(\rho A_k)
\bigr)\in\mathbb{R}^k
\;\big|\;
\rho\in \mathcal X
\Bigr\}
\label{eq:LX-def}
\end{equation}
is invariant under local unitary transformations of the observables, that is,
\begin{equation}
L_{\mathcal X}(U A_1 U^\dagger,\dots,U A_k U^\dagger)
=
L_{\mathcal X}(A_1,\dots,A_k),
\qquad U=U_A\otimes U_B.
\label{eq:LX-LU-invariance}
\end{equation}
In particular, Eq.~\eqref{eq:LX-LU-invariance} holds for the JNR,
where \(\mathcal X\) is the set of all quantum states, and for the JSNR,
where \(\mathcal X\) is the set of separable states.
\end{proposition}

\begin{proof}
By definition, we have
\begin{equation}
L_{\mathcal X}(U A_1 U^\dagger,\dots,U A_k U^\dagger)
=
\Bigl\{
\bigl(
\Tr(\rho\, U A_1 U^\dagger),\dots,\Tr(\rho\, U A_k U^\dagger)
\bigr)
\;\big|\;
\rho\in \mathcal X
\Bigr\}.
\label{eq:LX-UAU-def}
\end{equation}
Using the cyclicity of the trace, for each $i=1,\dots,k$,
\begin{equation}
\Tr(\rho\, U A_i U^\dagger)
=
\Tr(U^\dagger \rho U\, A_i).
\label{eq:trace-cyclicity}
\end{equation}
Since \(\mathcal X\) is invariant under local unitary conjugations as stated in
Eq.~\eqref{eq:X-LU-invariant}, the map
\begin{equation}
\rho \mapsto U^\dagger \rho U
\label{eq:rho-conjugation}
\end{equation}
is a bijection from \(\mathcal X\) onto itself. Therefore, as \(\rho\) ranges over \(\mathcal X\), so does \(U^\dagger \rho U\).
Combining Eqs.~\eqref{eq:LX-UAU-def}--\eqref{eq:rho-conjugation}, we obtain
\begin{equation}
L_{\mathcal X}(U A_1 U^\dagger,\dots,U A_k U^\dagger)
=
L_{\mathcal X}(A_1,\dots,A_k),
\label{eq:LX-equality-final}
\end{equation}
which coincides with Eq.~\eqref{eq:LX-LU-invariance}.
This completes the proof.
\end{proof}
Similarly, if \(\mathcal X\) in Proposition \ref{prop:LU-invariance-LX} represents all quantum states, then \(L_{\mathcal X}\) is (globally) unitary invariant.
In the following we analyze the geometric properties of the JSNR.
To simplify the analysis, we restrict ourselves to the case of two operators,
both of which are density operators of pure states.

\begin{proposition}
 \label{pro:JNR-two-pure}
 Let \(|\psi_1\rangle\) and \(|\psi_2\rangle\) be two pure states in a finite-dimensional
 Hilbert space \(\mathcal H\) with \(\dim \mathcal H \ge 2\), and let
 \(P_1 = |\psi_1\rangle\langle \psi_1|\) and \(P_2 = |\psi_2\rangle\langle \psi_2|\).
 Let \(x_1\) and \(x_2\) denote the two coordinates of the JNR or, in the product-state
 setting of (ii) below, the JSNR. For \(\gamma \in [0,1]\), define
 \[
 \mathcal{E}_\gamma
 :=
 \left\{ (x_1, x_2) \in [0, 1]^2 \mid (x_1 + x_2 - (1 - \gamma))^2 \le 4\gamma x_1 x_2 \right\}.
 \]
\begin{enumerate}
 \item
 Let \(c = |\langle \psi_1 | \psi_2 \rangle|^2\) and
 \(V = \mathrm{span}\{|\psi_1\rangle,|\psi_2\rangle\}\).
 The JNR of \(P_1\) and \(P_2\) is given by
 \begin{equation}
 \label{eq:conic}
 L(P_1, P_2)
 =
 \begin{cases}
 \mathcal{E}_c, & V^\perp = \{0\},\\[1mm]
 \mathrm{conv}\bigl(\mathcal{E}_c \cup \{(0,0)\}\bigr), & V^\perp \neq \{0\}.
 \end{cases}
 \end{equation}
 \item
 Suppose, in addition, that \(\mathcal H = \mathcal{H}_A \otimes \mathcal{H}_B\) is
 bipartite with \(\dim\mathcal H_A\ge2\) and \(\dim\mathcal H_B\ge2\), and that
 \(|\psi_1\rangle = |a_1\rangle \otimes |b_1\rangle\),
 \(|\psi_2\rangle = |a_2\rangle \otimes |b_2\rangle\) are product states, with local
 fidelities \(c_A = |\langle a_1 | a_2 \rangle|^2\) and
 \(c_B = |\langle b_1 | b_2 \rangle|^2\). Then the JSNR is given by
 \begin{equation}
 \label{eq:separable-JNR}
 L_{\mathrm{SEP}} (P_1, P_2)
 =
 \mathrm{conv} \Bigl(
 \Bigl\{
 \bigl(x_{1,A}x_{1,B},x_{2,A}x_{2,B}\bigr)
 \,\Big|\,
 (x_{1,A},x_{2,A})\in \mathcal{E}_{c_A}, \,
 (x_{1,B}, x_{2,B}) \in \mathcal{E}_{c_B}
 \Bigr\}
 \Bigr).
 \end{equation}
\end{enumerate}
\end{proposition}
\begin{proof}
 For part (i), Eq.~\eqref{eq:conic} is obtained by reducing the problem to the
 two-dimensional subspace \(V=\mathrm{span}\{|\psi_1\rangle,|\psi_2\rangle\}\)
 and parameterizing the pure states in \(V\); the detailed derivation is given
 in Appendix \ref{app:jnr_derivation}.

 For part (ii), Eq.~\eqref{eq:separable-JNR} is obtained by applying part (i)
 to the two local pairs and then taking the convex hull of the pointwise-product
 image; the detailed derivation is given in Appendix \ref{app:jsnr_general_derivation}.
\end{proof}

\begin{Rmk}
 Proposition \ref{prop:LU-invariance-LX} and Proposition \ref{pro:JNR-two-pure} reveal that the JNR
 for two pure states is determined solely by their
 fidelity \(c\), while the JSNR of two product states is
 determined by their local fidelities \(c_A, c_B\).
 This means that any two pairs of pure states
 with the same global fidelity \(c\) yield identical JNRs, while any two pairs of
 product states with the same local fidelities \(c_A,c_B\) yield identical JSNRs,
 regardless of the specific states or the dimensions of the underlying Hilbert spaces.
This property significantly simplifies the analysis of
entanglement detection using multiple fidelities, as it reduces
the JNR to a single parameter \(c\) and the JSNR of two product states
to the pair of local parameters \(c_A,c_B\).
\end{Rmk}

\begin{example}
 \label{ex:2qubit}
 Consider a 2-qubit system with two product states:
 \(|\psi_1\rangle = |00\rangle\) and \(|\psi_2\rangle = |++\rangle\), where
 \(|+\rangle = \frac{1}{\sqrt{2}}(|0\rangle + |1\rangle)\). The local fidelities are
 \(c_A = |\langle 0 | + \rangle|^2 = \frac{1}{2}\) and \(c_B = |\langle 0 | + \rangle|^2 = \frac{1}{2}\).
 Using Proposition \ref{pro:JNR-two-pure}, we can determine the JNR
 and JSNR for these two product states. The curved part of the JNR boundary is
 \begin{equation}
 \label{eq:analytical1}
 (x_1 + x_2 - \tfrac{3}{4})^2 = x_1 x_2, \quad x_1, x_2 \in [0, 1].
 \end{equation}
 For the JSNR, the common local range is
 \begin{equation}
 \mathcal{E}_{c_A} = \mathcal{E}_{c_B}
 =
 \left\{ (u, v)\in[0,1]^2 \,\middle|\, \left(u-\tfrac{1}{2}\right)^2 + \left(v-\tfrac{1}{2}\right)^2 \le \tfrac{1}{4}\right\}.
 \end{equation}
 Hence
 \[
 L_{\mathrm{SEP}}(P_1,P_2)
 =
 \mathrm{conv}\bigl\{(u_Au_B,v_Av_B)\mid (u_A,v_A),(u_B,v_B)\in \mathcal{E}_{c_A}\bigr\}.
 \]
In particular, the full JSNR contains the axis intercepts \((0,\tfrac{1}{2})\) and
\((\tfrac{1}{2},0)\); for instance, \((0,\tfrac{1}{2})\) is realized by the local
points \((x_{1,A},x_{2,A})=(0,\tfrac{1}{2})\in\mathcal{E}_{c_A}\) and
\((x_{1,B},x_{2,B})=(\tfrac{1}{2},1)\in\mathcal{E}_{c_B}\), while
\((\tfrac{1}{2},0)\) is realized by \((x_{1,A},x_{2,A})=(\tfrac{1}{2},0)\) and
\((x_{1,B},x_{2,B})=(1,\tfrac{1}{2})\). By contrast, the JNR conic \eqref{eq:analytical1} reaches the
larger intercepts \((0,\tfrac{3}{4})\) and \((\tfrac{3}{4},0)\), which already shows
that the JSNR is a strict subset of the JNR. The full JSNR is obtained numerically
by sampling the pointwise-product set and taking its convex hull.
Figure \ref{fig:JSNR_2qubit} illustrates the JNR and JSNR for these two product states,
highlighting the differences in their geometric structures.
\end{example}

\begin{figure}[t!]
 \centering
 \includegraphics[width=0.6\textwidth]{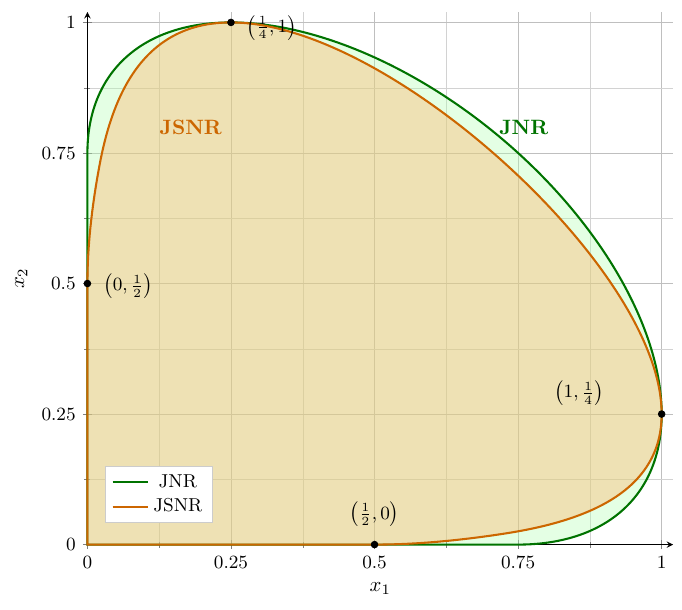}
 \caption{Comparison of the JNR and the JSNR for reference states $|\Psi_1\rangle = |00\rangle$ and $|\Psi_2\rangle = |++\rangle$, where $c = 0.25$ and $c_A = c_B = 0.5$. The JNR boundary is the conic \((x_1+x_2-\tfrac{3}{4})^2=x_1x_2\). The JSNR is obtained numerically as the convex hull of the pointwise-product set of the two local circular ranges. Although the JSNR is symmetric with respect to \(x_2 = x_1\), its boundary consists of piecewise segments rather than a single conic; in particular, it contains the intercepts \((0,\tfrac{1}{2})\) and \((\tfrac{1}{2},0)\).}
 \label{fig:JSNR_2qubit}
\end{figure}

\section{Conclusion}
\label{sec:con}
We have analyzed the entanglement detection power of multiple fidelities via the geometry
of the JSNR. By examining its supporting
hyperplanes, we derived necessary and sufficient conditions under which a set of \(k\)
product-state projectors forms an effective entanglement witness family.
Crucially, effectiveness is determined by whether at least one corresponding
linear combination has a maximal-eigenvalue eigenspace that is completely
entangled.

We further showed that if every pair of reference product states is either orthogonal or shares a local factor,
no linear combination can yield an effective witness unless the orthogonal complement of
the span of the whole set is a CES. UPBs provide a typical example of this mechanism:
if the reference states form a UPB, the full set is effective, while any proper subset is ineffective,
which differs from the case of just two product states.

Finally, we have shown that for two reference product states on a bipartite system
with arbitrary local dimensions, the geometry of their JSNR depends only on their local
fidelities. That is, for any two pairs of reference product states, if the
corresponding local fidelities are identical, then the resulting JSNRs have
the same geometry. In the representative two-qubit example, the full JSNR is
obtained by taking the convex hull of the pointwise-product set of the local
numerical ranges.

 Promising directions for future work include extending these methods to general
observables and multipartite systems, exploring connections with other entanglement
criteria, and optimizing measurement strategies for experiments.

\section*{Acknowledgments}

This work was supported by the National Natural Science Foundation of China under Grant Nos. 62072119.

\appendix
\section{Derivation of the JNR for two pure states}
\label{app:jnr_derivation}

In this appendix, we provide the formal derivation of the JNR for two pure state projectors $P_1 = |\psi_1\rangle\langle \psi_1|$ and $P_2 = |\psi_2\rangle\langle \psi_2|$. The JNR is defined as the set of expectation values attainable by all possible quantum states $\rho$ such that
\begin{equation}
L(P_1, P_2) = \{ (\mathrm{Tr}(\rho P_1), \mathrm{Tr}(\rho P_2)) \mid \rho \ge 0, \mathrm{Tr}(\rho) = 1 \}.
\end{equation}
By the properties of the numerical range, this set is convex and equivalent to the convex hull of the expectation values generated by pure states $|\phi\rangle$. The geometry of the set is uniquely determined by the fidelity $c = |\langle \psi_1 | \psi_2 \rangle|^2$.

 If \(c=1\), then \(P_1=P_2\) and
\[
L(P_1,P_2)=\{(t,t)\in[0,1]^2\}=\mathcal E_c,
\]
so part (i) of Proposition \ref{pro:JNR-two-pure} is immediate. We may therefore assume \(0\le c<1\) in the nontrivial derivation below.

 Let \(V=\mathrm{span}\{|\psi_1\rangle,|\psi_2\rangle\}\). Since both projectors are supported on \(V\), a general normalized pure state \(|\phi\rangle\in\mathcal H\) can be decomposed as
\begin{equation}
|\phi\rangle
=
\sqrt{t}\,|\phi_{\parallel}\rangle
+
\sqrt{1-t}\,|\phi_{\perp}\rangle,
\qquad t\in[0,1],
\end{equation}
where \(|\phi_{\parallel}\rangle\in V\), \(|\phi_{\perp}\rangle\perp V\), and both vectors are normalized whenever the corresponding coefficient is nonzero. Then
\[
\bigl(\langle \phi | P_1 | \phi \rangle,\langle \phi | P_2 | \phi \rangle\bigr)
=
t\bigl(\langle \phi_{\parallel} | P_1 | \phi_{\parallel} \rangle,\langle \phi_{\parallel} | P_2 | \phi_{\parallel} \rangle\bigr).
\]
Therefore it suffices to determine first the expectation-value set generated by normalized vectors inside \(V\).

To do so, after fixing the global phase of \(|\psi_2\rangle\) appropriately, we
 introduce an orthonormal basis $\{|e_1\rangle, |e_2\rangle\}$ of \(V\) such that
$|\psi_1\rangle = |e_1\rangle$ and $|\psi_2\rangle = \sqrt{c}|e_1\rangle +
\sqrt{1-c}|e_2\rangle$. Any normalized pure state $|\phi\rangle$ within this
subspace is parametrized by an angle $\theta \in [0, \pi/2]$ and a relative
phase $\varphi \in [0, 2\pi)$ as
\begin{equation}
|\phi\rangle = \cos\theta |e_1\rangle + \mathrm{e}^{\mathrm{i}\varphi} \sin\theta |e_2\rangle.
\end{equation}
The corresponding expectation values for the two projectors are given by the relations
\begin{align}
x_1 &= \langle \phi | P_1 | \phi \rangle = \cos^2\theta, \label{eq:x1_app} \\
x_2 &= \langle \phi | P_2 | \phi \rangle = c\cos^2\theta + (1-c)\sin^2\theta + 2\sqrt{c(1-c)}\cos\theta\sin\theta\cos\varphi. \label{eq:x2_app}
\end{align}
By substituting $x_1$ and $\sin^2\theta = 1-x_1$ into \eqref{eq:x2_app}, the expression for $x_2$ becomes
\begin{equation}
x_2 = cx_1 + (1-c)(1-x_1) + 2\sqrt{c(1-c)x_1(1-x_1)}\cos\varphi.
\end{equation}
The boundary of the set for this subspace is obtained at the extrema
 \(\cos\varphi = \pm 1\). Squaring the rearranged equation leads to the boundary equation
 \begin{equation}
(x_1 + x_2 - (1 - c))^2 = 4cx_1x_2, \label{eq:ellipse_app}
\end{equation}
which describes the boundary of the set \(\mathcal E_c\) generated inside \(V\).

 If \(V^\perp=\{0\}\), then every normalized pure state lies in \(V\), so the pure-state image is exactly \(\mathcal E_c\), and consequently
\[
L(P_1,P_2)=\mathcal E_c.
\]
If \(V^\perp\neq\{0\}\), then \((0,0)\) is attained by every normalized vector in \(V^\perp\), and every pure-state image point has the form \(t(y_1,y_2)\) with \((y_1,y_2)\in\mathcal E_c\) and \(t\in[0,1]\). Hence the pure-state image is contained in \(\mathrm{conv}(\mathcal E_c\cup\{(0,0)\})\). Conversely, \(\mathcal E_c\) and \((0,0)\) are both attained, so
\[
L(P_1,P_2)=\mathrm{conv}(\mathcal E_c\cup\{(0,0)\}).
\]

\section{Derivation of the JSNR for two product states}
\label{app:jsnr_general_derivation}

In this appendix, we derive Eq.~\eqref{eq:separable-JNR} for two product-state
projectors
\[
P_j = |a_j\rangle\langle a_j| \otimes |b_j\rangle\langle b_j|,
\qquad j=1,2.
\]
Since separable states are the convex hull of pure product states and the map
\(\rho \mapsto (\mathrm{Tr}(\rho P_1),\mathrm{Tr}(\rho P_2))\) is linear, it
suffices to characterize the image of pure product states.

For a pure product state \(|\alpha\rangle \otimes |\beta\rangle\), the expectation
values factorize as
\[
\bigl(\langle \alpha,\beta| P_1 |\alpha,\beta\rangle,\langle \alpha,\beta| P_2 |\alpha,\beta\rangle\bigr)
=
\bigl(u_{1,A}u_{1,B},u_{2,A}u_{2,B}\bigr),
\]
where
\[
u_{j,A} = |\langle a_j|\alpha\rangle|^2,
\qquad
u_{j,B} = |\langle b_j|\beta\rangle|^2,
\qquad j=1,2.
\]
Set \(u_A:= (u_{1,A},u_{2,A})\) and \(u_B:= (u_{1,B},u_{2,B})\). Applying
part (i) of Proposition \ref{pro:JNR-two-pure} to the two local pairs
\(|a_1\rangle,|a_2\rangle\) and \(|b_1\rangle,|b_2\rangle\), we obtain
\[
u_A \in \mathrm{conv}\bigl(\mathcal{E}_{c_A}\cup\{(0,0)\}\bigr),
\qquad
u_B \in \mathrm{conv}\bigl(\mathcal{E}_{c_B}\cup\{(0,0)\}\bigr),
\]
where the local fidelities are \(c_A = |\langle a_1|a_2\rangle|^2\) and
\(c_B = |\langle b_1|b_2\rangle|^2\).

Let
\[
S
=
\Bigl\{
\bigl(x_{1,A}x_{1,B},x_{2,A}x_{2,B}\bigr)
\,\Big|\,
(x_{1,A},x_{2,A}) \in \mathcal{E}_{c_A},\,
(x_{1,B},x_{2,B}) \in \mathcal{E}_{c_B}
\Bigr\}.
\]
The set \(S\) contains \((0,0)\): if \(c_A = 1\) or \(c_B = 1\), then
\((0,0)\) belongs to \(\mathcal{E}_{c_A}\) or \(\mathcal{E}_{c_B}\), respectively;
if \(c_A < 1\) and \(c_B < 1\), then
\((1-c_A,0) \in \mathcal{E}_{c_A}\) and \((0,1-c_B) \in \mathcal{E}_{c_B}\),
whose pointwise product is \((0,0)\).

Choose convex decompositions
\[
u_A = \sum_r \lambda_r x^{(r)},
\qquad
u_B = \sum_s \mu_s y^{(s)},
\]
with \(x^{(r)} \in \mathcal{E}_{c_A}\cup\{(0,0)\}\),
\(y^{(s)} \in \mathcal{E}_{c_B}\cup\{(0,0)\}\), \(\lambda_r \ge 0\),
\(\mu_s \ge 0\), and \(\sum_r \lambda_r = \sum_s \mu_s = 1\). Then
\[
\bigl(u_{1,A}u_{1,B},u_{2,A}u_{2,B}\bigr)
=
\sum_{r,s} \lambda_r \mu_s \bigl(x^{(r)}_1 y^{(s)}_1,x^{(r)}_2 y^{(s)}_2\bigr).
\]
If \(x^{(r)}=(0,0)\) or \(y^{(s)}=(0,0)\), the corresponding term equals
\((0,0)\in S\); otherwise it belongs to \(S\) by definition. Hence every pure-product
point lies in \(\mathrm{conv}(S)\).

Let \(V_A=\mathrm{span}\{|a_1\rangle,|a_2\rangle\}\) and
\(V_B=\mathrm{span}\{|b_1\rangle,|b_2\rangle\}\).

Conversely, if \(c_A<1\), then part (i) of Proposition \ref{pro:JNR-two-pure}
applied to the pair \(|a_1\rangle,|a_2\rangle\) shows that every point of
\(\mathcal{E}_{c_A}\) is attained by a normalized pure state in \(V_A\); if
\(c_A=1\), then \(\mathcal{E}_{c_A}=\{(t,t)\in[0,1]^2\}\), and each such point is
attained by a normalized pure state \(|\alpha\rangle\in\mathcal H_A\) with
\(|\langle a_1|\alpha\rangle|^2=t\). The same argument applies to
\(\mathcal{E}_{c_B}\). Therefore every point of \(S\) is attained by a pure
product state. It follows that the convex hull of the pure-product image is
exactly \(\mathrm{conv}(S)\), which proves Eq.~\eqref{eq:separable-JNR}.

\section*{References}
\bibliographystyle{iopart-num}
\bibliography{JPA}

\providecommand{\newblock}{}
\begin{thebibliography}{10}
\expandafter\ifx\csname url\endcsname\relax
  \def\url#1{{\tt #1}}\fi
\expandafter\ifx\csname urlprefix\endcsname\relax\def\urlprefix{URL }\fi
\providecommand{\eprint}[2][]{\url{#2}}
% Bibliography created with iopart-num v2.1
% /biblio/bibtex/contrib/iopart-num

\bibitem{horodecki2009}
Horodecki R, Horodecki P, Horodecki M and Horodecki K 2009 {\em Reviews of
  Modern Physics\/} {\bf 81} 865--942

\bibitem{nielsen2010quantum}
Nielsen M~A and Chuang I~L 2010 {\em Quantum computation and quantum
  information\/} (Cambridge University Press)

\bibitem{bennett1993}
Bennett C~H, Brassard G, Cr{\'e}peau C, Jozsa R, Peres A and Wootters W~K 1993
  {\em Physical Review Letters\/} {\bf 70} 1895--1899

\bibitem{bennett1992}
Bennett C~H and Wiesner S~J 1992 {\em Physical Review Letters\/} {\bf 69}
  2881--2884

\bibitem{bennett2014}
Bennett C~H and Brassard G 2014 {\em Theoretical Computer Science\/} {\bf 560}
  7--11

\bibitem{peres1996}
Peres A 1996 {\em Physical Review Letters\/} {\bf 77} 1413

\bibitem{chen2003}
Chen K and Wu L~A 2003 A matrix realignment method for recognizing entanglement
  (\textit{Preprint} \eprint{quant-ph/0205017})

\bibitem{rudolph2003}
Rudolph O 2003 {\em Physical Review A\/} {\bf 67} 032312

\bibitem{horodecki1997}
Horodecki P 1997 {\em Physics Letters A\/} {\bf 232} 333--339

\bibitem{chruscinski2014}
Chru{\'s}ci{\'n}ski D and Sarbicki G 2014 {\em Journal of Physics A:
  Mathematical and Theoretical\/} {\bf 47} 483001

\bibitem{guhne2009entanglement}
G{\"u}hne O and T{\'o}th G 2009 {\em Physics Reports\/} {\bf 474} 1--75

\bibitem{terhal2001family}
Terhal B~M 2001 {\em Linear Algebra and its Applications\/} {\bf 323} 61--73

\bibitem{lewenstein2000optimization}
Lewenstein M, Kraus B, Cirac J~I and Horodecki P 2000 {\em Physical Review A\/}
  {\bf 62} 052310

\bibitem{toth2005entanglement}
T{\'o}th G 2005 {\em Physical Review A—Atomic, Molecular, and Optical
  Physics\/} {\bf 71} 010301

\bibitem{piani2007class}
Piani M and Mora C~E 2007 {\em Physical Review A—Atomic, Molecular, and
  Optical Physics\/} {\bf 75} 012305

\bibitem{bourennane2004experimental}
Bourennane M, Eibl M, Kurtsiefer C, Gaertner S, Weinfurter H, G{\"u}hne O,
  Hyllus P, Bru{\ss} D, Lewenstein M and Sanpera A 2004 {\em Physical Review
  Letters\/} {\bf 92} 087902

\bibitem{weilenmann2020}
Weilenmann M, Dive B, Trillo D, Aguilar E~A and Navascu{\'e}s M 2020 {\em
  Physical Review Letters\/} {\bf 124} 200502

\bibitem{guhne2021}
G{\"u}hne O, Mao Y and Yu X~D 2021 {\em Physical Review Letters\/} {\bf 126}
  140503

\bibitem{zhang2025}
Zhang R and Wei Z 2025 {\em Quantum Science and Technology\/} {\bf 10} 015061

\bibitem{gutkin2013}
Gutkin E and {\.Z}yczkowski K 2013 {\em Linear Algebra and its Applications\/}
  {\bf 438} 2394--2404

\bibitem{bennett1999}
Bennett C~H, DiVincenzo D~P, Mor T, Shor P~W, Smolin J~A and Terhal B~M 1999
  {\em Physical Review Letters\/} {\bf 82} 5385

\bibitem{wu2020}
Wu P and Tang R 2020 {\em Journal of Physics A: Mathematical and Theoretical\/}
  {\bf 53} 445302

\bibitem{simnacher2021confident}
Simnacher T, Czartowski J {\em et~al.\/} 2021 {\em arXiv preprint
  arXiv:2107.04365\/}

\bibitem{bertsekas2003convex}
Bertsekas D, Nedic A and Ozdaglar A 2003 {\em Convex analysis and
  optimization\/} vol~1 (Athena Scientific)

\bibitem{parthasarathy2004}
Parthasarathy K~R 2004 On the maximal dimension of a completely entangled
  subspace for finite level quantum systems (\textit{Preprint}
  \eprint{quant-ph/0405077})

\bibitem{bengtsson2017}
Bengtsson I and {\.Z}yczkowski K 2017 {\em Geometry of quantum states: an
  introduction to quantum entanglement\/} (Cambridge University Press)

\end{thebibliography}

\end{document}